\let\accentvec\vec

\documentclass[runningheads,envcountsect,10pt]{llncs}

\let\vec\accentvec

\usepackage{amsmath}
\usepackage{mathrsfs}

\usepackage{multirow}
\usepackage{amsfonts}
\usepackage[colorlinks]{hyperref}
\usepackage{setspace}
\usepackage{textcomp}
\usepackage{enumerate}
\usepackage[ruled,linesnumbered]{algorithm2e}

\usepackage{graphicx}
\DeclareGraphicsExtensions{.pdf,.jpg,.png}

\begin{document}
%

\title{A True $O(n\log{n})$ Algorithm for the All-k-Nearest-Neighbors Problem}

\titlerunning{A True $O(n\log{n})$ Algorithm for the All-k-NN Problem}
%
\author{Hengzhao Ma\inst{1} \and
	Jianzhong Li\inst{2}
}
\authorrunning{H. Ma, J. Li}
%
\institute{
	\email{hzma@stu.hit.edu.cn}\and
	\email{lijzh@hit.edu.cn}\\
	Harbin Institute of Technology, Harbin, Heilongjiang 150001, China
}
\maketitle              
\begin{abstract}
	In this paper we examined an algorithm for the All-k-Nearest-Neighbor problem proposed in 1980s, which was claimed to have an $O(n\log{n})$ upper bound on the running time. We find the algorithm actually exceeds the so claimed upper bound, and prove that it has an $\Omega(n^2)$ lower bound on the time complexity. Besides, we propose a new algorithm that truly achieves the $O(n\log{n})$ bound. Detailed and rigorous theoretical proofs are provided to show the proposed algorithm runs exactly in $O(n\log{n})$ time.
	
\end{abstract}

\section{Introduction}\label{sec:intro}

The All-k-Nearest-Neighbors problem, or All-kNN for short, is an important problem that draws intensive research efforts. Early works about All-kNN date back to 1980s \cite{Clarkson1983}, and there are still some new results about this problem published in recent years \cite{Park2016,Sieranoja2015}. The reason why the All-kNN problem has been continuously studied is that many applications invoke All-kNN as an important sub-procedure, such as classification \cite{Szummer2001}, agglomerative clustering \cite{Franti2006}, image retrieval \cite{Yang2011}, recommendation systems \cite{Karypis2001},  and outlier detection \cite{Brito1997}. For many of these applications, solving All-kNN is reported as the main bottleneck \cite{Franti2006}.

The All-kNN problem can be briefly defined as follows. Let $(X,D)$ be a metric space, where $D(\cdot,\cdot)$ is a distance function and $X$ is a point set. 
The input of All-kNN is a point set $P\subseteq X$, and the output is the $k$-nearest-neighbor for all points $p\in P$, where the $k$-nearest-neighbor of a point $p$ is a set of $k$ points in $P$ that are closest to $p$ according to the distance function $D(\cdot,\cdot)$. 
The formal definition will be given in Section \ref{sec:def-prel}. 

There is an obvious brute-force solution for All-kNN, which is to compute the pairwise distances for all points in $P$, and select the $k$ points with smallest distance to $p$ for each point $p$. This solution takes $O(n^2)$ time, which is unacceptable when the size $n$ of the input is large. There have been a lot of algorithms proposed to efficiently solve the All-kNN problem, which can be categorized into the following three classes.
 
The first class of algorithms uses different techniques to accelerate the empirical running time of the algorithm, while the theoretical $O(n^2)$ upper bound is unchanged. 
There are basically three different kinds of techniques. The first kind is based on tree-like spacial indexes, such as $k$-$d$ trees \cite{Friedman1977} and Voronoi diagram \cite{Edelsbrunner2012} based index. The second kind is based on space filling curves, including Hilbert curve \cite{Hilbert1935} and Z-order curves \cite{Morton1966}. 
The space filling curve is a useful method to build an one-dimensional index on multidimensional data. There is an important property about the index based on the space filling curves, that is, the elements near to each other tend to be indexed into near entries. This property helps to reduce the number of distance computation to solve the All-kNN problem, as reported in \cite{Connor2010,Sieranoja2015,Yao2010}. The third kind is based on the idea of \emph{neighborhood propagation}, which uses the intuition that the neighbors of neighbors are also likely to be neighbors. The NN-descent algorithm proposed in \cite{Dong2011} is seminal work on neighborhood propagation, and it is still one of the best algorithms for All-kNN problem. Other works use the neighborhood propagation technique to refine the primary result returned by some preprocessing step. For example, the authors of \cite{Wang2012} use a multiple random divide and conquer approach as preprocessing, and the Locality Sensitive Hashing is used in \cite{Zhang2013}.

The second class of algorithms turns to solve the problem on parallel systems. 
The theoretical work \cite{Callahan1993} states that there is an optimal solution for All-kNN which needs $O(\log{n})$ time and $O(n)$ processors on CREW PRAM model. 
Other works tries to solve All-kNN on different parallel platforms, such as MapReduce platform \cite{Warashina2014,Trad2012} and GPU environment \cite{Komarov2013}.

Different from the most algorithms in the above two classes which do not reduce the theoretical $O(n^2)$ upper bound on running time, the fourth class of algorithms is proved to have lower upper bounds. And they are all serial algorithms, different with the work in \cite{Callahan1993}. For example, Bentley gives an multidimensional divide-and-conquer framework that can sole the All-kNN problem in $O(n(\log{n})^{d-1})$ time \cite{Bentley1980}, where $d$ is the number of dimension. 
Besides, the algorithm given in \cite{Clarkson1983} takes $O(n(\log{\delta}))$ time, where $\delta$ is the ratio of the maximum and minimum distance between any two points in the input. Finally, the algorithm proposed by Vaidya \cite{Vaidya1989} is claimed to have upper bound of $O(kd^dn\log{n})$ on the running time. 

After all, it can be summarized that most of the works about the All-kNN problem focus on improving the empirical running time of the algorithms, but few of them succeed in reducing the $O(n^2)$ worst case upper bound. To the best of our knowledge, the algorithms proposed by \cite{Bentley1980,Clarkson1983,Vaidya1989} are the only ones that have lower upper bound that $O(n^2)$. 
We will not consider the parallel situation so that the work in \cite{Callahan1993} is excluded. Unfortunately, it has been as long as 30 years since the theoretical results are published. There is an urgent demand to renew and improve these historical results, and that is exactly the main work in this paper.

Among the three theoretical works, the one proposed by Vaidya \cite{Vaidya1989} has the best upper bound on $n$, which is $O(kd^dn\log{n})$. We have carefully examined the algorithm and the proofs in \cite{Vaidya1989}, and unfortunately a major mistake is found, which is that some part of the algorithm proposed in \cite{Vaidya1989} actually exceeds the $O(n\log{n})$ bound. Since the work has been cited over 300 times combined with the conference version \cite{Vaidya1986}, it is necessary to point out the mistake and fix it. Here the contributions of this paper are listed as follows.

\begin{enumerate}
	\item We point out that the algorithm proposed in \cite{Vaidya1989} needs $\Omega(n^2)+kd^dn\log{n}$ time, contradicting with the claimed $O(n\log{n})$ upper bound.
	\item A modified algorithm is proposed whose running time is proved to be up-bounded by $O(k(k+(\sqrt{d})^d) \cdot n\log{n})$. The algorithm only applies to the Euclidean space.
	\item While the algorithm proposed in \cite{Vaidya1989} is designed for All-1NN and needs non-trivial modifications to generalize to All-kNN, the algorithm in this paper can directly solve the All-kNN problem for arbitrary integral value of $k$.
\end{enumerate}

The rest of the paper is organized as follows. In Section \ref{sec:def-prel} the definition of the problem and some prerequisite knowledge are introduced. Then the algorithm proposed in \cite{Vaidya1989} is described and analyzed in Section \ref{sec:former-alg}. The modified algorithm is proposed in Section \ref{sec:our-alg}, and the correctness and the upper bound on running time is formally proved in Section \ref{sec:analyz}. Finally we conclude the paper in Section \ref{sec:conc}.

\section{Problem Definitions and Preliminaries}\label{sec:def-prel}
The problem studied in this paper is the All-k-Nearest-Neighbors problem under Euclidean Space. In the following discussions we will assume that the input is a set $P$ of points where each $p\in P$ is a d-dimensional vector $(p^{(1)},p^{(2)},\cdots,p^{(d)})$. The distance between two points $p$ and $p'$ is measured by the Euclidean distance, which is $D(p,p')=\sqrt{\sum\limits_{i=1}^{d}{(p^{(i)}-p'^{(i)}})^2}$.
The formal definition of the All-k-Nearest-Neighbors problem is given below. 

\begin{definition}[k-NN]\label{def:knn}
	Given the input point set $P\subseteq R^d$ and a query point $q\in R^d$, define $kNN(q,P)$ to be the set of $k$ points in $P$ that are nearest to $q$. Formally,
	\begin{enumerate}
		\item $kNN(q,P)\subseteq P$, and $|kNN(q,P)|=k$;
		\item $D(p,q)\le D(p',q)$ for $\forall p\in kNN(q,P)$ and $\forall p'\in P\setminus kNN(q,P)$.
	\end{enumerate}
\end{definition}

\begin{definition}[All-k-Nearest-Neighbors, All-kNN]\label{def:all-knn}
	Given the input point set $P \subseteq R^d$, find $kNN(p,P\setminus \{p\})$ for all $ p\in P$.
\end{definition}

There is an equivalent definition for the All-kNN problem, that is the $k$-NN graph construction problem whose definition is given below. The $k$-NN graph is intuitively the graph representation of the result of All-kNN problem. In the rest of the paper we will use All-kNN to refer to the problem.
\begin{definition}[kNN-graph]\label{def:knn-graph}
	Given the input point set  $P \subseteq R^d$, the kNN-graph for $P$ is a directed graph $G=(V,E)$ satisfying:
	\begin{enumerate}
		\item $V=P$;
		\item For any two points $p_1,p_2\in V$, $(p_1,p_2)\in E$ iff $p_2\in kNN(p_1, P\setminus \{p_1\})$.
	\end{enumerate}
\end{definition}

Define a rectangle $\mathfrak{r}$ in $R^d$ to be the product of $d$ intervals, i.e., $I_1\times I_2\times\cdots\times I_d$, where each interval $I_i$ can be open, closed or semi-closed for $1\le i\le d$. 
For a rectangle $\mathfrak{r}$, let $Lmax(\mathfrak{r}), Lmin(\mathfrak{r})$ be the longest and shortest length of the intervals defining $\mathfrak{r}$, respectively. When $Lmax(\mathfrak{r})=Lmin(\mathfrak{r})$, $\mathfrak{r}$ is called a d-cube, and denote the side length of the d-cube $\mathfrak{r}$ as $Len(\mathfrak{r})=Lmax(\mathfrak{r})=Lmin(\mathfrak{r})$. For a point set $P$, $\mathfrak{r}$ is called the bounding rectangle if $\mathfrak{r}$ is the smallest rectangle containing all the points in $P$, then let $|\mathfrak{r}|$ to be equivalent to $|P|$ that is the number of points in $P$. Besides, a d-cube $\mathfrak{r}$ is a Minimal Cubical Rectangle (MCR) for a given point set $P$ iff $\mathfrak{r}$ contains all the points in $P$ and has the minimal side length. Note that for a specific point set $P$, its bounding rectangle is unique but its MCR may not.

On the other hand, define a d-ball to be the set $B(c,r)=\{x\in R^d\mid D(x,c)\le r \}$, where $c$ is the center and $r$ is the radius of the d-ball. For a point set $P$, define the Minimum Enclosing Ball (MEB) of $P$ to be the minimum radius ball containing all the points in $P$, denoted as $MEB(P)$. It is known that the there exists one unique $MEB(P)$ for a given $P$, which can be computed by solving a quadratic-programming problem \cite{Yildirim2008}. From now on let $\mathcal{C}_P$ and $\mathcal{R}_P$ denote the unique center and radius of $MEB(P)$ respectively. Besides the exact MEB, the approximate MEB is equally useful and easier to compute. A d-ball $B(c_P,r_P)$ is an $\epsilon$-MEB of a point set $P$ iff $P\subseteq B(c_P,r_P)$ and $r_P\le \epsilon\cdot \mathcal{R}_P$. The following algorithm can compute an $\frac{3}{2}$-MEB for a given set $P$ in linear time, which is proposed in \cite{Zarrabi-Zadeh2006}. 

\begin{algorithm}[H]
	\caption{Compute $\frac{3}{2}$-MEB}\label{alg:approx-meb}
	\KwIn{$B(c_P,r_P)$ which is a $\frac{3}{2}$-MEB of $P$, and another point set $Q$}
	\KwOut{a $\frac{3}{2}$-MEB of $P\cup Q$}
	\uIf{$(c_P,r_P)==Null$}{
		$c_0\gets$ a random point in $Q$\;
		$r_0\gets 0$\;		
	}\Else{
		$c_0\gets c_P, r_0\gets r_p$\;
	}
	\While{$\exists q\in Q$}{
		\If{$d(q,c_0)> r_0$}{
			$\delta\gets \frac{1}{2}(D(q,c_0)-r_0)$\;
			$r_1\gets r_0+\delta$\;
			$c_1\gets c_0+\frac{\delta}{D(q,c_0)}(q-c_0)$\;
			$c_0\gets c_1,r_0\gets r_1$\;
		}
		$P\gets P\cup \{q\}, Q\gets Q\setminus\{q\}$\;
	}
\end{algorithm}

The main characteristic of Algorithm \ref{alg:approx-meb} is that it can be viewed as a dynamic algorithm, where the $\frac{3}{2}$-MEB of $P$ is precomputed, and the set $Q$ is an update to $P$. The algorithm runs in $O(|Q|)$ time to compute the $\frac{3}{2}$-MEB of $P\cup Q$ based on the precomputed $\frac{3}{2}$-MEB of $P$. This characteristic will play an important part in the algorithm proposed in Section \ref{sec:our-alg}.

\section{The Algorithm in \cite{Vaidya1989} and analysis of it}\label{sec:former-alg}
In this section we cite the outlines of the algorithm proposed in \cite{Vaidya1989}, and give the analysis why this algorithm is not an $O(n\log{n})$ algorithm. Note that the algorithm focuses on the situation that $k=1$. The generalization to $k>1$ is non-trivial but the author did not give the generalized algorithm in detail in \cite{Vaidya1989}. This is a disadvantage of the algorithm, as was pointed out in Section \ref{sec:intro}.

\subsection{The algorithm}\label{subsec:its-alg}
The algorithm mainly consists of two parts, i.e., building the Cube Split Tree, and maintain the $Nbr$ and $Frd$ sets for each node in the cube split tree. The definitions of the Cube Split Tree, $Nbr$ and $Frd$ sets are given in the discussions below.

\subsubsection{The Cube Split Tree.}
Given the input point set $P$, the algorithm will start from the $MCR$ of $P$, split it into a set of d-cubes, and organize the d-cubes into a tree structure, which is called the Cube Split Tree (CST). Let $Succ(\mathfrak{r})$ denote the set of the sub-rectangles generated by splitting $\mathfrak{r}$, then the CST can be defined as follows.

\begin{definition}[Cube Split Tree, CST]\label{def:cst}
	Given a point set $P$, a CST based on $P$ is a tree structure $T$ satisfying the following properties:
	\begin{enumerate}
		\item the root of $T$ is $MCR(P)$,
		\item each of the nodes of $T$ represents a d-cube, 
		which is an MCR of a subset of $P$, and
		\item there is an edge between $\mathfrak{r}$ and each $\mathfrak{r}'\in Succ(\mathfrak{r})$.
		
	\end{enumerate}
	Besides, $T$ is called fully built if all the leaf nodes contains only one point.
\end{definition}

The structure of the CST is determined by the method to split a rectangle $\mathfrak{r}$ and generate $Succ(\mathfrak{r})$. Now we introduce the method described in \cite{Vaidya1989}. Given a d-cube $\mathfrak{r}$, let $\alpha(\mathfrak{r})$ denote the geometric center of $\mathfrak{r}$. Then let $h_i(\mathfrak{r})$ denote the hyperplane passing through $\alpha(\mathfrak{r})$ and orthogonal to the $i$-th coordinate axis, $1\le i\le d$. The $d$ hyperplanes divide $\mathfrak{r}$ into $2^d$ cells. Discard the cells containing no points, shrink the non-empty cells to MCR's and $Succ(\mathfrak{r})$ is obtained. 
The following lemma shows the property of the CST constructed by using the above split method. The detailed proof can be found in \cite{Vaidya1989}.

\begin{lemma}[\cite{Vaidya1989}]\label{lema:its-time-split}
	Given a point set $P$ where $|P|=n$, the CST $T$ built based on $P$ has the following properties:
	\begin{enumerate}
		\item there are at most $2n$ d-cubes in $T$, and
		\item the total time to conduct the above split method and build a fully-built CST is $O(2^d n\log{n})$.
	\end{enumerate}

\end{lemma}

\subsubsection*{The $Nbr$ and $Frd$ sets.}
For each rectangle $\mathfrak{r}$ in the CST $T$, two sets $Nbr(\mathfrak{r})$ and $Frd(\mathfrak{r})$ will be maintained, whose definitions are given below.

For any two rectangles  $\mathfrak{r}$ and $\mathfrak{r}'$ in the CST $T$, define $D_{max}(\mathfrak{r}), D_{min}(\mathfrak{r},\mathfrak{r}'),D_{max}(\mathfrak{r},\mathfrak{r}')$ as follows:

$$D_{max}(\mathfrak{r})=\max\limits_{p,p'\in \mathfrak{r}}\{D(p,p')\}$$
$$D_{min}(\mathfrak{r},\mathfrak{r}')=\min\limits_{p\in \mathfrak{r},p'\in \mathfrak{r}'}\{D(p,p')\}$$
$$D_{max}(\mathfrak{r}, \mathfrak{r}')=\max\limits_{p\in \mathfrak{r},p'\in \mathfrak{r}'}\{D(p,p')\}$$
Then define $Est(\mathfrak{r})$ by the following equation:

\begin{equation}\label{eqtn:est(b)}
Est(\mathfrak{r})= \left\{
\begin{aligned}
D_{max}(\mathfrak{r}) &,& if\;|\mathfrak{r}|\ge 2\\
\min\limits_{\mathfrak{r}'\in T} \{D_{max}(\mathfrak{r},\mathfrak{r}')\}&,& otherwise
\end{aligned}
\right.
\end{equation}

Now the definitions of the $Nbr$ and $Frd$ sets can be given using the denotations above.

\begin{definition}
	$Nbr(\mathfrak{r})=\{\mathfrak{r}'\in T \mid D_{min}(\mathfrak{r},\mathfrak{r}')\le Est(\mathfrak{r})\}$.
\end{definition}

\begin{definition}
	$Frd(\mathfrak{r})=\{\mathfrak{r}'\in T\mid \mathfrak{r}\in Nbr(\mathfrak{r}') \}$.
\end{definition}

The algorithm will ensure that  $kNN(p,P\setminus\{p\})\subseteq Nbr(\mathfrak{r})$ for $\forall p\in \mathfrak{r}$. When the algorithm terminates and each rectangle $\mathfrak{r}$ contains only one point, there will be $Nbr(\mathfrak{r})=kNN(p,P\setminus \{p\})$ where $p$ is the only one point in $\mathfrak{r}$. 

The following two lemma shows the size of $Nbr$ and $Frd$ sets maintained for each rectangle in the CST $T$.

\begin{lemma}[\cite{Vaidya1989}]\label{lema:nbr-size-cite}
	$|Nbr(\mathfrak{r})|\le 2^d(2d+3)^d$, for $\forall \mathfrak{r}\in T$.
\end{lemma}

\begin{lemma}[\cite{Vaidya1989}]\label{lema:frd-size-cite}
	$|Frd(\mathfrak{r})|\le 12^d+2^d(8d+3)^d$, for $\forall \mathfrak{r}\in T$.
\end{lemma}

\subsubsection*{The overall algorithm.}
Now we are ready to give the overall algorithm proposed in \cite{Vaidya1989}. The pseudo codes are divided into the main procedure and the $MntnNbrFrd$ sub-procedure.

\begin{algorithm}[H]
	\caption{The algorithm}\label{alg:cite}
	\KwIn{A point set $P$}
	\KwOut{The result of All-kNN}
	\SetKwFunction{MntnNbrFrd}{MntnNbrFrd}
	\SetKwProg{Procedure}{Procedure}{\string:}{end}
	$\mathfrak{r}_0\gets MCR(P)$, $Nbr(\mathfrak{r}_0)\gets\emptyset$, $Frd(\mathfrak{r}_0)\gets\emptyset$, $\mathcal{S}\gets \{\mathfrak{r}_0\}$\;
	Create a tree rooted at $\mathfrak{r}_0$\;
	\While{$|\mathcal{S}|\le|P|$}{
		$\mathfrak{r}\gets$ the rectangle in $\mathcal{S}$ that has the largest volume\;
		Split $\mathfrak{r}$ and generates $Succ(\mathfrak{r})$\;
		\MntnNbrFrd{$\mathfrak{r},Succ(\mathfrak{r})$}\;
		$\mathcal{S}\gets \mathcal{S}\setminus \{\mathfrak{r}\}\cup Succ(\mathfrak{r})$\;
	}
\end{algorithm}

\begin{algorithm}[thb]
	\caption{The $MntnNbrFrd$ sub-procedure}\label{alg:mntn-nbr-frd}
	\SetKwProg{Procedure}{Procedure}{\string:}{end}
	\SetKwFunction{MntnNbrFrd}{MntnNbrFrd}
	\Procedure{\MntnNbrFrd{$\mathfrak{r},Succ(\mathfrak{r})$}}{
		Create a node for $\mathfrak{r}'$ and hang it under the node for $\mathfrak{r}$ in the CST $T$;
		\ForEach{$\mathfrak{r}'\in Succ(\mathfrak{r})$}{
			$Nbr(\mathfrak{r}')\gets Nbr(\mathfrak{r})\cup Succ(\mathfrak{r})\setminus \{\mathfrak{r}'\}$\;
			$Frd(\mathfrak{r}')\gets Frd(\mathfrak{r})\cup Succ(\mathfrak{r})\setminus \{\mathfrak{r}'\}$\;
			Create $Est(\mathfrak{r}')$ according to Equation \ref{eqtn:est(b)}\;
		}
		\ForEach{$\mathfrak{r}'\in Frd(\mathfrak{r})$}{
			$Nbr(\mathfrak{r}')\gets Nbr(\mathfrak{r})\cup Succ(\mathfrak{r})\setminus \{\mathfrak{r}\}$\;
			$Est(\mathfrak{r})\gets \min\{Est(\mathfrak{r}'), \min\limits_{\mathfrak{r}''\in Succ(\mathfrak{r})} \{D_{max}(\mathfrak{r},\mathfrak{r}'')\} \}$\;
		}
		\ForEach{$\mathfrak{r}'\in Nbr(\mathfrak{r})$}{
			$Frd(\mathfrak{r}')\gets Frd(\mathfrak{r}')\cup Succ(\mathfrak{r})\setminus \{\mathfrak{r} \} $\;
		}
		\ForEach{$\mathfrak{r}'\in Succ(\mathfrak{r}\cup Frd(\mathfrak{r}))$}{
			\ForEach{$\mathfrak{r}''\in Nbr(\mathfrak{r}')$}{
				\If{$D_{min}(\mathfrak{r},\mathfrak{r}''> Est(\mathfrak{r}') )$}{
					Delete $\mathfrak{r}''$ from $Nbr(\mathfrak{r}')$\;
					Delete $\mathfrak{r}'$ from $Frd(\mathfrak{r}'')$\;
				}
			}
			
		}
	}
\end{algorithm}

It can be deduced from Algorithm \ref{alg:mntn-nbr-frd} that when the $MntnNbrFrd$ sub-procedure is invoked, there are at most $2|Succ(\mathfrak{r})| (|Succ(\mathfrak{r})|+|Nbr(\mathfrak{r})|+|Frd(\mathfrak{r}))|$ additions to the $Nbr$ and $Frd$ sets. On the other hand, the $while$ loop in Algorithm \ref{alg:cite} is executed for at most $2n$ times, because there are at most $2n$ d-cubes in the CST $T$ and each execution of the $while$ loop visits at least one d-cube in $T$. And thus the $MntnNbrFrd$ sub-procedure is invoked for at most $2n$ times. Combing with Lemma \ref{lema:nbr-size-cite} and \ref{lema:frd-size-cite}, the following lemma can be easily proved.

\begin{lemma}[\cite{Vaidya1989}]\label{lema:its-time-mntn-nbr-frd}
	The total time to maintain the $Nbr$ and $Frd$ sets, which is the total time to execute Algorithm \ref{alg:mntn-nbr-frd}, is $O(d^d\cdot n)$.
\end{lemma}

\subsection{Analysis of Algorithm \ref{alg:cite}}\label{subsec:its-analyz}
The total time to execute Algorithm \ref{alg:cite} can be divided into three parts, i.e., the time to conduct the split method and build the split tree, the time to maintain the $Nbr$ and $Frd$ sets, and  the time to maintain the $Est$ value for each d-cube. The former two parts are already solved by Lemma \ref{lema:its-time-split} and \ref{lema:its-time-mntn-nbr-frd}. For the last part, the author of \cite{Vaidya1989} claims that the total time to maintain the $Est$ value is $O(n)$, but this is impossible. The following lemma shows the impossibility.

\begin{lemma}\label{lema:its-time-est}
	The time to create and maintain the $Est(\mathfrak{r})$ value for all $\mathfrak{r}$ in the CST $T$ is $\Omega(n^2)$.
\end{lemma}

\begin{proof}
	We recite the definition of $Est(\mathfrak{r})$ here, where
	\begin{equation*}
	Est(\mathfrak{r})= \left\{
	\begin{aligned}
	D_{max}(\mathfrak{r}) &,& if\;|\mathfrak{r}|\ge 2\\
	\min\limits_{\mathfrak{r}'\in Nbr(\mathfrak{r})} \{D_{max}(\mathfrak{r},\mathfrak{r}')\}&,& otherwise
	\end{aligned}
	\right.
	\end{equation*}
	
	Please pay attention to the first term. When $|\mathfrak{r}|\ge 2$, $Est(\mathfrak{r})$ is defined by $Est(\mathfrak{r})=D_{max}(\mathfrak{r})$, 
	where $D_{max}(\mathfrak{r})=\max\limits_{p_1,p_2\in \mathfrak{r}}{D(p_1,p_2)}$. We assure that this is the original definition in \cite{Vaidya1989}. According to this definition, it will exceed the $O(n\log{n})$ bound merely to compute $Est(\mathfrak{r}_0)$, 
	where $\mathfrak{r}_0$ is the MCR of the input point set $P$. Actually, it needs $\Theta(n^2)$ time to compute $Est(\mathfrak{r}_0)$, because it needs to compute the pairwise distance between each pair of points in $P$, and there are obviously $\frac{n(n-1)}{2}$ pairs. Then the time to compute $Est(\mathfrak{r})$ for all $\mathfrak{r}$ in the CST $T$ is obviously $\Omega(n^2)$.
	\qed
\end{proof}

Finally, we give the following theorem stating a lower bound of the time complexity of Algorithm \ref{alg:cite}, contradicting with the analysis given in \cite{Vaidya1989}.

\begin{theorem}
	An lower bound of the time complexity of Algorithm \ref{alg:cite} is $\Omega(n^2)$.
\end{theorem}

\section{The algorithm in this paper}\label{sec:our-alg}
According to the analysis in Section \ref{subsec:its-analyz}, the main defect of Algorithm \ref{alg:cite} is that the $Est$ value defined in Equation \ref{eqtn:est(b)} can not be efficiently computed. In this section we propose our algorithm that solves this problem elegantly. 

The algorithm is divided into tree parts, which will be introduced one by one in the rest of this section.

\subsection{Constructing the Rectangle Split Tree}
The first part of our algorithm is to build the Rectangle Split Tree (RST), which is very similar to the algorithm for building the CST. There are several differences between the two algorithms. First, our algorithm represents the subsets of the input point set $P$ by bounding rectangles, rather than MCR's. Second, each time when a rectangle $\mathfrak{r}$ is chosen, it is split into two sub-rectangles by cutting the longest edge of $\mathfrak{r}$ into two equal halves. Let $\mathfrak{r}_{large}$ denote the one in the two sub-rectangles of $\mathfrak{r}$ that contains more points, and $\mathfrak{r}_{small}$ denotes the other.  The Rectangle Split Tree is defined in the following Definition 4.1.

\begin{definition}[Rectangle Split Tree, RST]\label{def:rst}
	Given a point set $P$, a RST based on $P$ is a tree structure satisfying:
	\begin{enumerate}
		\item the root of $T$ is the bounding rectangle of $P$,
		\item each node in $T$ represents a rectangle, which is the bounding rectangle of a subset of $P$, and
		\item there is an edge from $\mathfrak{r}$ to $\mathfrak{r}_{large}$ and $\mathfrak{r}_{small}$.
	\end{enumerate}
	Besides, $T$ is called fully built if all the leaf nodes contain only one point.
\end{definition}

Next we give the algorithm to build the RST. The complexity of Algorithm \ref{alg:rst} will be proved to be $O(dn\log{n})$ in Section \ref{sec:analyz}.

\begin{algorithm}[H]
	\caption{Constructing the RST}\label{alg:rst}
	\KwIn{a point set $P$, and the bounding rectangel $\mathfrak{r}_P$  of $P$}
	\KwOut{an RST $T$}
	Create a tree $T$ rooted at $\mathfrak{r}_P$\;
	$\mathcal{S}_0\gets \{\mathfrak{r}\},\mathcal{S}_1\gets\emptyset$\;
	
	\While{$|\mathcal{S}_1|<|P|$}{
		$\mathfrak{r}\gets$ an arbitrary rectangle in $\mathcal{S}_0$\;
		Split $\mathfrak{r}$ into $\mathfrak{r}_{large}$ and $\mathfrak{r}_{small}$\;
		
		\ForEach{$\mathfrak{r}_i\in \{ \mathfrak{r}_{large}, \mathfrak{r}_{small}\}$}{
			Create a node for $\mathfrak{r}_i$ and hang it under the node for $\mathfrak{r}$ in $T$\;
			
			\uIf{$|\mathfrak{r}_i|>1$}{
				$\mathcal{S}_0\gets\mathcal{S}_0\cup \{\mathfrak{r}_i\}$\;
			}\Else{
				$\mathcal{S}_1\gets\mathcal{S}_1\cup \{\mathfrak{r}_i\}$\;
			}
		}
		$\mathcal{S}_0\gets \mathcal{S}_0\setminus \{\mathfrak{r}\}$\;
	}
\end{algorithm}

\subsection{Computing the Approximate MEB}
\begin{algorithm}[H]
	\SetKwFunction{ComputeMES}{ComputeMES}
	\SetKwProg{Procedure}{Procedure}{\string:}{end}
	\caption{Compute the approximate MES}\label{alg:compute-mes}
	\KwIn{a RST $T$}
	\KwOut{Compute an $\frac{3}{2}$-MES for $\forall \mathfrak{r}\in T$}
	Invoke \ComputeMES($root(T)$), where $root(T)$ is the root of $T$\;
	\Procedure{\ComputeMES{$\mathfrak{r}$}}{
		\If{$|\mathfrak{r}|=1$}{
			$c_\mathfrak{r}=p$, where $p$ is the only point in $\mathfrak{r}$\;
			$r_{\mathfrak{r}}=0$\;
			return;
		}
		\ComputeMES{$\mathfrak{r}_{large}$}\;
		\ComputeMES{$\mathfrak{r}_{small}$}\;		
		Invoke Algorithm \ref{alg:approx-meb}, where the parameters are set to : $c_{\mathfrak{r}_{large}},r_{\mathfrak{r}_{large}},\mathfrak{r}_{small}$\;
	}
\end{algorithm}

The second part of the proposed algorithm is to compute the approximate MES for each node in the RST, which is given as Algorithm \ref{alg:compute-mes}. This algorithm receives the constructed RST $T$ as input, and traverse $T$ with post-root order, where Algorithm \ref{alg:approx-meb} will be invoked at each node. It will be shown in Section \ref{sec:analyz} that the algorithm takes $O(dn\log{n})$ time.

\subsection{Computing All-kNN}\label{subsec:all-knn}
Based on the algorithm for constructing the RST and computing MES, the algorithm for All-kNN is given as Algorithm \ref{alg:all-knn}. It is worthy to point out that the algorithm naturally applies to all integer $k\ge 1$, which is an advantage against Algorithm \ref{alg:cite}.

The algorithm first invoke Algorithm \ref{alg:rst} on $P$ to construct an $RST$ $T$ (Line \ref{line:aknn:calling-rst}) . Then  Algorithm \ref{alg:compute-mes} is invoked at Line \ref{line:aknn:calling-mes} to compute the approximate MES for each node in $T$. The rest of the algorithm aims to traverse $T$ and construct $kNbr(\mathfrak{r})$ and $kFrd(\mathfrak{r})$ sets for each $\mathfrak{r}\in T$. The formal definition of the two sets along with two auxiliary definitions are given below. In these definitions, assume that a set $H$ of rectangles is given, and the approximate MES $(c_{\mathfrak{r}},r_{\mathfrak{r}})$ of each $\mathfrak{r}\in H$ is precomputed.

\begin{definition}\label{def:qlty}
	Given two rectangles $\mathfrak{r}$ and $\mathfrak{r}'\in H$, define the relative quality of $\mathfrak{r}'$ against $\mathfrak{r}$ as follows: 
	$$Qlty(\mathfrak{r}',\mathfrak{r})=\sqrt{(D(c_{\mathfrak{r}'},c_{\mathfrak{r}})+\frac{1}{2}r_{\mathfrak{r}'})^2+r_{\mathfrak{r}'}^2 }$$
	
\end{definition}

\begin{definition} Let $\min\limits_{k}\{x\in S\mid f(x) \}$ be the k-th smallest value in the set $\{x\in S\mid f(x) \}$. Define $kThres(\mathfrak{r})$ for each $\mathfrak{r}\in H$ as follow: 
	\begin{equation}\label{eqtn:kthres}
	kThres(\mathfrak{r})=\left\{
	\begin{aligned}
	2\cdot r_{\mathfrak{r}}&,& if |\mathfrak{r}|\ge k+1\\
	\min\limits_{k}\left\{\mathfrak{r}'\in H\setminus\{\mathfrak{r}\} \mid Qlty(\mathfrak{r}',\mathfrak{r})  \right\}&,& if |\mathfrak{r}|<k+1\\
	\end{aligned}	
	\right.
	\end{equation}
	
\end{definition}

\begin{definition}\label{def:knbr-set}
	$kNbr(\mathfrak{r})=\{\mathfrak{r}'\in H \mid D(c_{\mathfrak{r}'},c_{\mathfrak{r}})\le r_{\mathfrak{r}}+r_{\mathfrak{r}'}+kThres(\mathfrak{r})  \}$.
	
\end{definition}

\begin{definition}\label{def:kfrd-set}
	$kFrd(\mathfrak{r})=\{ \mathfrak{r}'\in H \mid \mathfrak{r}\in kNbr(\mathfrak{r}') \} $.
\end{definition}

Next we describe how Algorithm \ref{alg:all-knn} works to construct the desired $kNbr$ and $kFrd$ sets.

The algorithm will visit all the rectangles in the RST $T$ by a descending order on the radius of the approximate MES. To do so, a heap $H$ is maintained to store the rectangles in $T$ ordered by the radius of the approximate MES of them. The main part of Algorithm \ref{alg:all-knn} is a $While$ loop. Each time the $While$ loop is executed, the top element $\mathfrak{r}$ of $H$, which is the one with the largest approximate MES radius, will be popped out of $H$. Then the algorithm will determine the $\mathcal{S}_{son}$ set, push all the rectangles in $\mathcal{S}_{son}$ into $H$ and process the rectangles in $\mathcal{S}_{son}$ by invoking the $MntnNbrFrd$ process (Algorithm \ref{alg:new-mntn-nbrfrd}) . The $While$ loop will terminate when $|H|=|P|$ and that is when all the rectangles in $T$ are processed.

The set $\mathcal{S}_{son}$ is determined by the following criterion.  
First, $\mathcal{S}_{son}$ is set to be $\{\mathfrak{r}_{large},\mathfrak{r}_{small} \}$. 
Then, if any $\mathfrak{r}_i\in \mathcal{S}_{son}$ satisfies $|\mathfrak{r}_i|<k+1$, $\mathfrak{r}_i$ will be replaced by the set of the leaf nodes in the subtree of $T$ rooted at $\mathfrak{r}_i$. Such work is done by Line \ref{line:aknn:determine-son-start} to Line \ref{line:aknn:determine-son-end} in Algorithm \ref{alg:all-knn}.

\begin{algorithm}[H]
	\caption{All-kNN}\label{alg:all-knn}
	\KwIn{a point set $P$}
	\KwOut{All-kNN on $P$}
	\SetKwFunction{MntnNbrFrd}{MntnNbrFrd}
	
	Invoke Algorithm \ref{alg:rst} on $P$, and construct an RST $T$\;\label{line:aknn:calling-rst}
	Invoke \ComputeMES{$root(T)$} where $root(T)$ is the root of $T$\;\label{line:aknn:calling-mes}
	
	$kNbr(root(T))=\emptyset, kFrd(root(T))=\emptyset$\;
	
	Initialize a heap $H=\{root(T)\}$, which is order by the radius of the approximate MES\;
	\While{$|H|<|P|$}{\label{line:aknn:while-loop}
		$\mathfrak{r}\gets H.pop()$\;\label{line:aknn:pop}
		$\mathcal{S}_{son}\gets \{\mathfrak{r}_{large},\mathfrak{r}_{small}\}$\;
		
		\ForEach{$\mathfrak{r}_i\in \mathcal{S}_{son}$}{
			\If{$|\mathfrak{r}_i|<k+1$}{\label{line:aknn:determine-son-start}
				
				$\mathcal{S}_{leaf}\gets $ the set of the leaf nodes in the subtree rooted at $\mathfrak{r}_i$\;
				$\mathcal{S}\gets \mathcal{S}_{son}\setminus \{\mathfrak{r}_i\}\cup \mathcal{S}_{leaf}$\;
				Delete the non-leaf nodes in the sub-tree\;\label{line:aknn:deleting-smallers}
				\label{line:aknn:final-delete-inner}
				
				Hang the leaf nodes directly under $\mathfrak{r}$\;
				
			}
			Push $\mathfrak{r}_i$ into $H$\;
		}\label{line:aknn:determine-son-end}
		
		\MntnNbrFrd{$\mathfrak{r}, \mathcal{S}_{son}$}\;
	}

\end{algorithm}

After the $\mathcal{S}_{son}$ is determined, Algorithm \ref{alg:all-knn} invokes the $MntnNbrFrm$ process to construct the $kNbr$ and $kFrd$ sets for each $\mathfrak{r}_i\in \mathcal{S}_{son}$. Though the algorithm is named the same with Algorithm \ref{alg:mntn-nbr-frd}, they actually work very differently. Algorithm \ref{alg:new-mntn-nbrfrd} relies on four sub-procedures, which are $AddInNbr$, $DelFromNbr$, $DelFromFrd$ and $TruncateNbr$, respectively. The name of these sub-procedures intuitively shows their functionality, and the pseudo codes are given in Algorithm \ref{alg:sub-proc}. Back to the $MntnNbrFrd$ process, Algorithm \ref{alg:new-mntn-nbrfrd} conducts the following four steps.

\begin{enumerate}[Step 1.]
	\item In Lines \ref{line:aknn:iterate-nbr-start} to \ref{line:aknn:iterate-nbr-end}, the algorithm iterates over $kNbr(\mathfrak{r})$. For each $\mathfrak{r}'\in kNbr(\mathfrak{r})$, $\mathfrak{r}$ will be deleted from $kFrd(\mathfrak{r}')$ since $\mathfrak{r}$ is split and replaced by $\mathcal{S}_{son}$. Then the algorithm invokes the $AddInNbr$ sub-procedure to try to add $\mathfrak{r}'$ into $kNbr(\mathfrak{r}_i)$ for each $\mathfrak{r}_{i}\in \mathcal{S}_{son}$. It can be seen from the pseudo codes of the $AddInNbr$ sub-procedure that if $\mathfrak{r}'$ is added into $kNbr(\mathfrak{r}_i)$ then $\mathfrak{r}_i$ will be added into $kFrd(\mathfrak{r}')$.
	\item In Lines \ref{line:aknn:iterate-frd-start} to \ref{line:aknn:iterate-frd-end}, the algorithm iterates over $kFrd(\mathfrak{r})$. For each $\mathfrak{r}'\in kFrd(\mathfrak{r})$, the algorithm will delete $\mathfrak{r}$ from $kNbr(\mathfrak{r}')$ and invoke $AddInNbr$ to try to add each $\mathfrak{r}_i\in \mathcal{S}_{son}$ into $kNbr(\mathfrak{r}')$.
	\item In Lines \ref{line:aknn:iterate-son-start} to \ref{line:aknn:iterate-son-end}, the algorithm iterates over $\mathcal{S}_{son}$. For each $\mathfrak{r}_i\in \mathcal{S}_{son}$, the $AddInNbr$ sub-procedure will be invoked to try adding each rectangle $\mathfrak{r}_j\in \mathcal{S}_{son}\setminus \{\mathfrak{r}_i \}$ into $kNbr(\mathfrak{r}_i)$.
	\item In Lines \ref{line:aknn:truncate-start} to \ref{line:aknn:truncate-end}, the $TruncateNbr$ sub-procedure is invoked on each $\mathfrak{r}'\in \mathcal{S}_{son}\cup kFrd(\mathfrak{r})$ where $|\mathfrak{r}'|=1$. The word \emph{truncate} indicates that it will delete a section of $kNbr(\mathfrak{r})$, as will be explained next.
\end{enumerate}

\begin{algorithm}[H]
	\caption{The $MntnNbrFrd$ process}\label{alg:new-mntn-nbrfrd}
	\SetKwFunction{DelFromFrd}{DelFromFrd}
	\SetKwFunction{DelFromNbr}{DelFromNbr}
	\SetKwFunction{AddInNbr}{AddInNbr}
	\SetKwFunction{TruncateNbr}{TruncateNbr}
	
	\SetKwProg{Procedure}{Procedure}{\string:}{end}
	\SetKwFunction{MntnNbrFrd}{MntnNbrFrd}
	
	\Procedure{\MntnNbrFrd{$\mathfrak{r},\mathcal{S}_{son}$}}{
		\ForEach{$\mathfrak{r}'\in kNbr(\mathfrak{r})$}{\label{line:aknn:iterate-nbr-start}
			\DelFromFrd{$\mathfrak{r},\mathfrak{r}'$}\tcp*{Delete $\mathfrak{r}$ from $kFrd(\mathfrak{r}')$}\label{line:aknn:del-from-frd}
			\ForEach{$\mathfrak{r}_i\in \mathcal{S}_{son}$}{
				\AddInNbr($\mathfrak{r}',\mathfrak{r}_i$)\tcp*{Add $\mathfrak{r}'$ into $kNbr(\mathfrak{r}_i)$}
			}
		}\label{line:aknn:iterate-nbr-end}
		
		\ForEach{$\mathfrak{r}'\in kFrd(\mathfrak{r})$}{\label{line:aknn:iterate-frd-start}
			\DelFromNbr{$\mathfrak{r},\mathfrak{r}'$}\tcp*{Delete $\mathfrak{r}$ from $kNbr(\mathfrak{r}')$}\label{line:aknn:del-from-nbr}
			\ForEach{$\mathfrak{r}_i\in \mathcal{S}_{son}$}{
				\AddInNbr($\mathfrak{r}_i,\mathfrak{r}'$)\tcp*{Add $\mathfrak{r}_i$ into $kNbr(\mathfrak{r}')$}
			}
		}\label{line:aknn:iterate-frd-end}
		
		\ForEach{$\mathfrak{r}_i\in\mathcal{S}_{son}$}{\label{line:aknn:iterate-son-start}
			\For{$\mathfrak{r}_j\in \mathcal{S}_{son}\setminus\{\mathfrak{r}_i \}$}{
				\AddInNbr{$\mathfrak{r}_i,\mathfrak{r}_j$}\tcp*{Add $\mathfrak{r}_i$ into $kNbr(\mathfrak{r}_j)$}
				
			}
		}\label{line:aknn:iterate-son-end}

		\ForEach{$\mathfrak{r}'\in \mathcal{S}_{son}\cup kFrd(\mathfrak{r})$}{\label{line:aknn:truncate-start}
			\If{$|\mathfrak{r}'|=1$}{
				\TruncateNbr{$\mathfrak{r}'$}\;
			}
			
		}\label{line:aknn:truncate-end}
	}

\end{algorithm}

Algorithm \ref{alg:sub-proc} shows the functionalities of the sub-procedures  mentioned in Algorithm \ref{alg:new-mntn-nbrfrd}, including $DelFromNbr$, $DelFromFrd$, $AddInNbr$ and $TruncateNbr$. 
Among them, the $AddInNbr$ sub-procedure needs to be explained in more details. 
It can be seen that an extra set $Cand(\mathfrak{r})$ is used in the special case of $|\mathfrak{r}|=1$, which stores the candidate rectangles that might contain the k-NN of the only point in $\mathfrak{r}$. The $Cand$ set is formally defined as follows.

\begin{definition}\label{def:cand-set}
	The $Cand(\mathfrak{r})$ set maintained for $\mathfrak{r}$ where $|\mathfrak{r}|=1$, is a set satisfies:
	\begin{enumerate}
		\item $Cand(\mathfrak{r})\subseteq kNbr(\mathfrak{r})$,
		\item $|Cand(\mathfrak{r})|=k$, and
		\item $Qlty(\mathfrak{r}',\mathfrak{r}) \le Qlty(\mathfrak{r}'',\mathfrak{r})$ for $\forall \mathfrak{r}'\in Cand(\mathfrak{r})$ and $\forall \mathfrak{r}''\in kNbr(\mathfrak{r})$.
	\end{enumerate}
\end{definition}

Recall that $kThres(\mathfrak{r})$ is defined to be $	\min\limits_{k}\{\mathfrak{r}'\in H\setminus\{\mathfrak{r}\} \mid Qlty(\mathfrak{r}',\mathfrak{r})\}$ when $|\mathfrak{r}|<k+1$. And then with $Cand(\mathfrak{r})$ defined, $kThres(\mathfrak{r})$ can be equivalently defined as $\max\limits_{\mathfrak{r}'\in Cand(\mathfrak{r}) }\{ Qlty(\mathfrak{r}',\mathfrak{r}) \}$, as is shown at Line \ref{line:aknn:kthres-another-def} in Algorithm \ref{alg:sub-proc}.

\begin{algorithm}[H]
	\caption{The sub-procedures}\label{alg:sub-proc}
	\SetKw{Continue}{continue}
	\SetKw{Break}{break}
	\SetKwFunction{DelFromFrd}{DelFromFrd}
	\SetKwFunction{DelFromNbr}{DelFromNbr}
	\SetKwFunction{AddInNbr}{AddInNbr}
	\SetKwFunction{TruncateNbr}{TruncateNbr}
	\SetKwProg{Procedure}{Procedure}{\string:}{end}
	
	\Procedure{\DelFromNbr{$\mathfrak{r}',\mathfrak{r}$}}{
		Delete $\mathfrak{r}'$ from $kNbr(\mathfrak{r})$\;
	}
	\Procedure{\DelFromFrd{$\mathfrak{r}',\mathfrak{r}$}}{
		Delete $\mathfrak{r}'$ from $kFrd(\mathfrak{r})$\;
	}
	\Procedure{\AddInNbr{$\mathfrak{r}',\mathfrak{r}$}}{
		\If{$|\mathfrak{r}|=1$}{
			\uIf{$|Cand(\mathfrak{r})|<k$}{
				Add $\mathfrak{r}'$ into $Cand(\mathfrak{r})$\;
			}\Else{
				\If{$Qlty(\mathfrak{r}',\mathfrak{r})<kThres(\mathfrak{r})$}{
					Add $\mathfrak{r}'$ into $Cand(\mathfrak{r})$\;
					Pop the top element out of $Cand(\mathfrak{r})$\;
					$kThres(\mathfrak{r})\gets \max\limits_{\mathfrak{r}''\in Cand(\mathfrak{r}) }\{ Qlty(\mathfrak{r}'',\mathfrak{r}) \}$\;\label{line:aknn:kthres-another-def}
				}
			}
			
		}\
		\If{$D(c_{\mathfrak{r}'},c_{\mathfrak{r}}) <r_{\mathfrak{r}'}+r_{\mathfrak{r}}+kThres(\mathfrak{r})$}{
			Add $\mathfrak{r}'$ into $kNbr(\mathfrak{r})$\;
			Add $\mathfrak{r}$ into $kFrd(\mathfrak{r}')$\;
		}
		
	}

	\Procedure{\TruncateNbr{$\mathfrak{r}$}}{
		Find the first element $\mathfrak{r}'\in kNbr(\mathfrak{r})$ such that $D(c_{\mathfrak{r}'},c_{\mathfrak{r}})-r_{\mathfrak{r}'}  >r_{\mathfrak{r}}+kThres(\mathfrak{r})$\;
		Truncate $kNbr(\mathfrak{r})$ by deleting all elements from $\mathfrak{r}'$ to the last one\;
	}
\end{algorithm}

Last but not least, these sub-procedures above require the $kNbr$, $kFrd$ and $Cand$ sets to be stored in specific data structures. The requirements and the data structures are discussed below.
\begin{enumerate}
	\item The data structure for the $kNbr$ sets should support insertion and deletion both efficiently, and it should be an ordered set to support the $TruncateNbr$ procedure. It is possible to use a B+ tree to support such operations efficiently. It is known that insertion and deletion of B+ tree takes $O(\log{n})$ time. For the operation of truncating, it is also possible to delete a section of elements in B+ tree in $O(\log{n})$ amortized time, regardless of the number of elements deleted.
	Besides, the $TruncateNbr$ sub-procedure requires the elements $\mathfrak{r}'\in kNbr(\mathfrak{r})$ being ordered by the value of $D(c_{\mathfrak{r}}',c_{\mathfrak{r}})-r_{\mathfrak{r}'}$. Because only such ordering can ensure the truncating operation would delete each rectangle $\mathfrak{r}'$ that satisfies $D(c_{\mathfrak{r}'},c_{\mathfrak{r}})>r_{\mathfrak{r}'}+r_{\mathfrak{r}}+kThres(\mathfrak{r})$ and keep the $kNbr(\mathfrak{r})$ set conforming to its definition.
	\item The operations on $kFrd$ sets include insertion and deletion by specific key. It is appropriate to implement it by red-black trees, which takes $O(\log{n})$ time for both insertion and deletion. There is no specific requirement on how to order the elements $\mathfrak{r}'\in kFrd(\mathfrak{r})$, and an natural choice is to use $r_{\mathfrak{r}'}$ as the key to order them.
	\item  For the $Cand$ set, the operations on it include pushing, popping, and finding the maximum value. Thus it can be implemented by a heap, and the elements $\mathfrak{r}'\in Cand(\mathfrak{r})$ should be ordered by $Qlty(\mathfrak{r}',\mathfrak{r})$.
\end{enumerate}

The correctness and complexities of the above proposed algorithms will be proved in the next section.

\section{Analysis}\label{sec:analyz}

\subsection{Correctness}

\begin{lemma}\label{lema:nbr-frd-subset-heap}
	Each time when the $while$ loop begins, Algorithm \ref{alg:all-knn} ensures the following two invariants:
	\begin{enumerate}
		\item $kNbr(\mathfrak{r})\subseteq H$ and $kFrd(\mathfrak{r})\subseteq H$ for $\forall \mathfrak{r}\in H$, and 
		\item if $\mathfrak{r}'\in kNbr(\mathfrak{r})$, then $\mathfrak{r}\in kFrd(\mathfrak{r}')$, for $\forall \mathfrak{r},\mathfrak{r}'\in H$.
	\end{enumerate}
\end{lemma}

\begin{proof}
	We prove the two statements by induction on the $while$ loop.
	
	Before the work-flow goes into the $while$ loop, $H$ contains only the root of $T$, and $kNbr(root(T))=kFrd(root(T))=\emptyset$. Then it can be easily verified that the two statements in the lemma are true at this time.
	
	As induction hypothesis, suppose that the two statements are true before the $i$-th execution of the $while$ loop starts.
	The $while$ loop body only changes the $kNbr(\mathfrak{r}')$ and $kFrd(\mathfrak{r}')$ sets of $\mathfrak{r}'\in kNbr(\mathfrak{r})\cup kFrd(\mathfrak{r})\cup \mathcal{S}_{son}$, where $\mathfrak{r}$ is the top element of $H$. First we consider the rectangles $\mathfrak{r}'\in kNbr(\mathfrak{r})\cup kFrd(\mathfrak{r})$ .
	It can be seen that $\mathfrak{r}$ is popped out of $H$ (Line \ref{line:aknn:pop}), and $\mathfrak{r}$ is deleted from $kNbr(\mathfrak{r}')$ and $kFrd(\mathfrak{r}')$ (Line \ref{line:aknn:del-from-frd} and \ref{line:aknn:del-from-nbr}). This operation does not violate the first statement. On the other and, the boxes in $\mathcal{S}_{son}$ are pushed into $H$, so adding $\mathfrak{r}_i\in \mathcal{S}_{son}$ into $kNbr(\mathfrak{r}')$ or $kFrd(\mathfrak{r}')$ does not violate the first statement. Thus the first statement holds for $\forall \mathfrak{r}'\in kNbr(\mathfrak{r})\cup kFrd(\mathfrak{r})$. Then for $\forall \mathfrak{r}_i\in \mathcal{S}_{son}$, the rectangles in $kNbr(\mathfrak{r}_i)\cup kFrd(\mathfrak{r}_i)$ are either inherited from $kNbr(\mathfrak{r})\cup kFrd(\mathfrak{r})$, which are subsets of $H$ according to the induction hypothesis, or other rectangles in $\mathcal{S}_{son}$ that are added into $H$ in this $while$ loop. Thus the first statement holds for $\mathfrak{r}_i\in \mathcal{S}_{son}$. Since the execution of the while loop does not influence the $kNbr$ and $kFrd$ sets of other rectangles, we arrive at the conclusion that the first statement holds for $\forall \mathfrak{r}\in H$ before the next $while$ loop starts.
	
	The second statement is ensured by the $AddInNbr$ and $TruncateNbr$ sub-procedures. It can be seen that adding (deleting) $\mathfrak{r}$ into $kNbr(\mathfrak{r}')$ always appears along with adding (or deleting) $\mathfrak{r}'$ into $kFrd(\mathfrak{r})$. Thus the second statement can be easily proved to be true.
	\qed
\end{proof}

\begin{lemma}\label{lema:cand-is-cand}
	The $kNbr(\mathfrak{r})$, $kFrd(\mathfrak{r})$ and $Cand(\mathfrak{r})$ sets constructed in Algorithm \ref{alg:sub-proc} meets their definition which are given in Definition \ref{def:knbr-set}, \ref{def:kfrd-set} and \ref{def:cand-set}.
\end{lemma}

\begin{proof}
	We write this as a lemma to ensure the rigorousness of the whole proof. Actually it can be easily verified and thus the proof is omitted.
\end{proof}

\begin{lemma}[\cite{Badoiu2003}]\label{lema:exists-point-pythagorean}
	Given a point set $P\subseteq R^d$ whose MEB is $B(\mathcal{C}_P,\mathcal{R}_P)$, then for any point $q\in R^d$, there exists a point $p\in P$ such that $D(q,p)\le \sqrt{D(q,\mathcal{C}_P)^2+\mathcal{R}_P^2}$.
\end{lemma}

\begin{proof}
	The proof can be found in \cite{Badoiu2003}.
\end{proof}

\begin{lemma}\label{lema:exists-point-qlty}
	Given a point set $P\subseteq R^d$ along with its $\frac{3}{2}$-MEB $B(c_P,r_P)$, then for any point $q\in R^d$, there exists a point $p\in P$ such that $D(q,p)\le \sqrt{(D(q,c_P)+\frac{1}{2}r_P)^2+r_P^2}$.
\end{lemma}

\begin{proof}
	Since $B(c_P,r_P)$ is the  $\frac{3}{2}$-MEB of $P$, we have $\mathcal{R}_P\le r_P\le \frac{3}{2}\mathcal{R}_P$, where $\mathcal{R}_P$ is the radius of the exact MEB of $P$. Pick a point $p_0$ such that $D(\mathcal{C}_P,p_0)=\mathcal{R}_P$ (such point definitely exists), then 
	$$D(c_P,\mathcal{C}_P)\le D(c_P,p_0)-D(p_0,\mathcal{C}_P)\le r_P-\mathcal{R}_P\le \frac{1}{2}\mathcal{R}_P\le \frac{1}{2}r_P.$$ 
	Thus, $$D(q,\mathcal{C}_P)\le D(q,c_P)+D(c_P,\mathcal{C}_P)\le D(q,c_P)+\frac{1}{2}r_P.$$
	Now just choose the point $p$ stated in Lemma \ref{lema:exists-point-pythagorean}, we have 
	$$D(q,p)\le \sqrt{D(q,\mathcal{C}_P)^2+\mathcal{R}_P^2}\le  \sqrt{(D(q,c_P)+\frac{1}{2}r_P)^2+r_P^2 }.$$
\end{proof}

\begin{lemma}\label{lema:tk-kthres}
	For any point $p\in P$, define $T_k(p)=\max\limits_{p'\in kNN(p,P\setminus\{p \} )}\{D(p,p') \}$. Then $T_k(p)\le kThres(\mathfrak{r})$ holds for any $p\in \mathfrak{r}$.
\end{lemma}

\begin{proof}
	The proof considers the non-leaf nodes and leaf nodes separately.
	
	For a non-leaf node $\mathfrak{r}$, the number of points inside $\mathfrak{r}$ is at least $k+1$ since the rectangles containing fewer points are already deleted by Line \ref{line:aknn:deleting-smallers}. On the other side, $kThres(\mathfrak{r})$ is set to $2 r_{\mathfrak{r}}$ according to Equation \ref{eqtn:kthres}. Note that for any point $p\in \mathfrak{r}$, the distance from $p'\in \mathfrak{r}\setminus\{p\}$ to $p$ is at most $2\mathcal{R}_{\mathfrak{r}}\le 2r_{\mathfrak{r}}$. Then it can be deduced that there is at least $k+1$ points within the range of $kThres(\mathfrak{r})$ around $p$. Since there is exactly $k$ points within the range of $T_k(p)$ around $p$, we get the conclusion that $T_k(p)\le kThres(\mathfrak{r})$  for a non-leaf node $\mathfrak{r}$ and $\forall p\in \mathfrak{r}$.
	
	For a leaf node $\mathfrak{r}$, first there is only one point $c_{\mathfrak{r}}\in \mathfrak{r}$. According to Lemma \ref{lema:exists-point-qlty}, there exist a point $p'\in \mathfrak{r}'$ at a distance of at most $\sqrt{(D(c_{\mathfrak{r}'},c_{\mathfrak{r}})^2+\frac{1}{2}r_{\mathfrak{r}'})+r_{\mathfrak{r}'}^2}$ from $c_{\mathfrak{r}}$ for each rectangle $\mathfrak{r}'\in kNbr(\mathfrak{r})$. Note that this is exactly the $Qlty(\mathfrak{r}',\mathfrak{r})$ defined in Definition \ref{def:qlty}. According to Definition \ref{def:cand-set} for the $Cand(\mathfrak{r})$ set, it stores the $k$ rectangles with the smallest $Qlty(\mathfrak{r}',\mathfrak{r})$ value. 
	Equivalently, there exists at least $k$ points in $Cand(\mathfrak{r})$. On the other hand, there are exactly $k$ points within the range of $T_k(p)$ around $p$. Thus it is easy to see that $T_k(p)\le \max\limits_{\mathfrak{r}'}\{Qlty(\mathfrak{r}',\mathfrak{r})\}=kThres(\mathfrak{r})$, where $\mathfrak{r}$ is a leaf node and $p=c_{\mathfrak{r}}$ is the only point in $\mathfrak{r}$.
	\qed
\end{proof}

\begin{lemma}\label{lema:knn-subset-knbr}
	$kNN(p, P\setminus \{p\})\subseteq \{\mathfrak{r}\}\cup kNbr(\mathfrak{r})$ for any rectangle $\mathfrak{r}\in T$ and any point $p\in \mathfrak{r}$.
\end{lemma}

\begin{proof}

	Given an arbitrary rectangle $\mathfrak{r}\in T$ and an arbitrary point $p\in \mathfrak{r}$, the following proofs applies to an arbitrary point $q\in kNN(p,P\setminus\{p\})$. First we have $D(q,p)\le T_k(p)$ since $q$ is one of the k-NN's of $p$. If $q\in \mathfrak{r}$ then the lemma trivially holds. Otherwise $p$ lies in another rectangle $\mathfrak{r}'$, which indicates that $D(q,c_{\mathfrak{r}'})\le r_{\mathfrak{r}'}$. On the other hand, $D(p,c_{\mathfrak{r}})\le r_{\mathfrak{r}}$ since $p\in \mathfrak{r}$. Thus, the following inequality can be derived based on the triangle inequality and Lemma \ref{lema:tk-kthres}:
	$$D(c_{\mathfrak{r}},c_{\mathfrak{r}'})\le D(c_{\mathfrak{r}},p)+D(p,q)+D(q,c_{\mathfrak{r}'})\le r_{\mathfrak{r}}+T_k(p)+r_{\mathfrak{r}'}\le r_{\mathfrak{r}}+kThres(\mathfrak{r})+r_{\mathfrak{r}'}.$$
	
	This inequality indicates that $\mathfrak{r}'\in kNbr(\mathfrak{r})$ according to Definition \ref{def:knbr-set}.
	
	Finally, it can be concluded that $kNN(p, P\setminus \{p\})\subseteq \{\mathfrak{r}\}\cup kNbr(\mathfrak{r})$ since the point $p$ is arbitrary in $kNN(p,P\setminus\{p\})$. 
	\qed
\end{proof}

By now we have proved that $Cand(\mathfrak{r})$ and $kNN(p, P\setminus \{p\})$ are both subsets of $kNbr(\mathfrak{r})$. Next we will show that the three sets collapse to the same one when Algorithm \ref{alg:all-knn} terminates, which fulfills the proof of the correctness of the algorithm that it can correctly compute $kNN(p, P\setminus \{p\})$ for all point $p\in P$. This 

\begin{lemma}\label{lema:knbr-minus-cand-non-leaf}
	Let $\mathfrak{r}$ to be a leaf node in $T$, then all $\mathfrak{r}'\in kNbr(\mathfrak{r})\setminus Cand(\mathfrak{r})$ are non-leaf nodes.
	
\end{lemma}

\begin{proof}
	Recall the definitions of $kNbr$ and $Cand$ sets in Definition \ref{def:knbr-set} and \ref{def:cand-set}. Suppose to the contrary that there exists $\mathfrak{r}'\in kNbr(\mathfrak{r})\setminus Cand(\mathfrak{r})$ and $\mathfrak{r}'$ is a leaf node that contains only one point. In such case,  $Qlty(\mathfrak{r}',\mathfrak{r})=D(c_{\mathfrak{r}'},c_{\mathfrak{r}})$ since $r_{\mathfrak{r}'}=0$. 
	Then $Qlty(\mathfrak{r}',\mathfrak{r})=D(c_{\mathfrak{r}'},c_{\mathfrak{r}}) > kThres(\mathfrak{r})$ because $\mathfrak{r}'\notin Cand(\mathfrak{r})$. On the other hand, $D(c_{\mathfrak{r}'},c_{\mathfrak{r}}) \le r_{\mathfrak{r}'}+r_{\mathfrak{r}}+kThres(\mathfrak{r})$ according to $\mathfrak{r}'\in kNbr(\mathfrak{r})$ and the definition of $kNbr(\mathfrak{r})$. Further, the last inequality can be simplified to $D(c_{\mathfrak{r}'},c_{\mathfrak{r}})\le kThres(\mathfrak{r})$ since $r_{\mathfrak{r}}=r_{\mathfrak{r}'}=0$. Thus, we get $D(c_{\mathfrak{r}'},c_{\mathfrak{r}})\le kThres(\mathfrak{r}) <D(c_{\mathfrak{r}'},c_{\mathfrak{r}})$, which is a contradiction. Finally, it is proved that $\mathfrak{r}'$ is a non-leaf node.
\end{proof}

\begin{theorem}\label{thrm:correctness-cand-equal-knn}
	When Algorithm \ref{alg:all-knn} terminates, $kNbr(\mathfrak{r})=Cand(\mathfrak{r})=kNN(p,P\setminus\{p\})$ holds for any leaf node $\mathfrak{r}$, where $p=c_{\mathfrak{r}}$ is the only point in $\mathfrak{r}$.
\end{theorem}

\begin{proof}
	When Algorithm \ref{alg:all-knn} terminates, all rectangles contains only one point, which indicates that $r_{\mathfrak{r}}=0$ for all $\mathfrak{r}$. 
	In such case, $Qlty(\mathfrak{r}',\mathfrak{r})$ is simplified to the form of $D(c_{\mathfrak{r}'},c_{\mathfrak{r}})$. And then it can be verified that the definition of $Cand(\mathfrak{r})$ is simplified to the same form with $kNN(c_{\mathfrak{r}},P\setminus\{c_{\mathfrak{r}}\})$. On the other hand, Lemma \ref{lema:cand-is-cand} ensures that the $Cand(\mathfrak{r})$ set constructed in the algorithm meets its definition, then $Cand(\mathfrak{r})=kNN(c_{\mathfrak{r}},P\setminus\{c_{\mathfrak{r}}\})$ is proved.
	
	Lemma \ref{lema:knbr-minus-cand-non-leaf} shows all $\mathfrak{r}$ in $kNbr(\mathfrak{r})\setminus Cand(\mathfrak{r})$ are non-leaf nodes. Since all $\mathfrak{r}\in H$ are leaf nodes when Algorithm \ref{alg:all-knn} terminates, and  Lemma  \ref{lema:nbr-frd-subset-heap} shows that $kNbr(\mathfrak{r})\subseteq H$, there will be $kNbr(\mathfrak{r})\setminus Cand(\mathfrak{r})=\emptyset$ which is equivalent to $kNbr(\mathfrak{r})=Cand(\mathfrak{r})$.
	
	After all, it is proved that $kNbr(\mathfrak{r})=Cand(\mathfrak{r})=kNN(c_{\mathfrak{r}},P\setminus\{c_{\mathfrak{r}}\})$, i.e., the three sets collapse to the same one.
	\qed
\end{proof}

\subsection{Complexities}

\subsubsection{5.2.1 Complexity of Algorithm \ref{alg:rst}}

\begin{lemma}\label{lema:number-of-nodes-rst}
	There are at most $2n$ nodes in the RST $T$.
\end{lemma}

\begin{proof}
	The construction of $T$ indicates that $T$ is binary but may not be a full binary tree. It is a known result that a full binary tree with $n$ leaf nodes contains $2n-1$ nodes, and thus the number of nodes in $T$ is no more than $2n$.
	\qed
\end{proof}

\begin{definition}
	Given a RST $T$ and any $\mathfrak{r}\in T$, call $\mathfrak{r}_{large}$ and  $\mathfrak{r}_{small}$ as the larger and smaller son of $\mathfrak{r}$, respectively. Then define $lsa(\mathfrak{r})$ (lowest smaller ancestor) for $\mathfrak{r}$ as follows. Let $\mathcal{P}$ be the path from $\mathfrak{r}$ to the root of $T$. If all the rectangles in $\mathcal{P}$ are the larger son of its farther, then $lsa(\mathfrak{r})$ is the root. Otherwise $lsa(\mathfrak{r})$ is the lowest rectangle in $\mathcal{P}$ that is the smaller son of its parent.
\end{definition}

\begin{lemma}[\cite{Vaidya1989}]\label{lema:split-step-time}
	For any $\mathfrak{r}$ in the RST $T$, the process of splitting $\mathfrak{r}$ into $\mathfrak{r}_{large}$ and $\mathfrak{r}_{small}$ can be done in $O(d\cdot |\mathfrak{r}_{small}|\cdot (1+\log{|lsa(\mathfrak{r})|}+\log{|\mathfrak{r}_{small}|} ) )$ time.
\end{lemma}

\begin{lemma}[\cite{Vaidya1989}]\label{lema:split-steps-sum}
	Let $w(\mathfrak{r})=d\cdot |\mathfrak{r}_{small}|\cdot (1+\log{|lsa(\mathfrak{r})|}+\log{|\mathfrak{r}_{small}|}) $. Then 
	$\sum\limits_{\mathfrak{r}\in T}{w(\mathfrak{r})}=O(dn\log{n})$. 
\end{lemma}

The detailed proof of Lemma \ref{lema:split-step-time} and \ref{lema:split-steps-sum} can be found in \cite{Vaidya1989}. And the following lemma can be directly deduced by the above two lemmas.

\begin{theorem}\label{thrm:time-alg-rst}
	The time complexity of Algorithm \ref{alg:rst} is $O(dn\log{n})$.
\end{theorem}

\subsubsection{5.2.2 Complexity of Algorithm \ref{alg:compute-mes}}

\begin{theorem}\label{thrm:time-alg-mes}
	The time complexity of Algorithm \ref{alg:compute-mes} is $O(dn\log{n})$.
\end{theorem}

\begin{proof}
	According to Algorithm \ref{alg:compute-mes}, the time to compute the $\frac{3}{2}$-MES of any  rectangle $\mathfrak{r}$ is $O(d|\mathfrak{r}_{small}|)$, which is strictly less than $w(\mathfrak{r})$. Then The complexity of Algorithm \ref{alg:compute-mes} can be expressed by $\sum\limits_{\mathfrak{r}\in T}{O(|\mathfrak{r}_{small}|) }$.
	Since Lemma \ref{lema:split-steps-sum} shows that $\sum\limits_{\mathfrak{r}\in T}{w(\mathfrak{r})}\le O(dn\log{n})$, it can be concluded that the time complexity of Algorithm \ref{alg:compute-mes} is $O(dn\log{n})$.	
	\qed
\end{proof}

\subsubsection{5.2.3 Complexity of Algorithm \ref{alg:all-knn}}

\begin{lemma}\label{lema:lmax-radius}
	For any point set $P$, let $\mathfrak{r}$ be the bounding rectangle of $P$, and let $B(\mathcal{C}_{P},\mathcal{R}_P)$ be the MES of $P$, then $Lmax(\mathfrak{r})\le 2\cdot \mathcal{R}_P\le \sqrt{d}\cdot Lmax(\mathfrak{r})$, where $Lmax(\mathfrak{r})$ is the length of the longest side of $\mathfrak{r}$.
\end{lemma}

\begin{proof}
	Let $\mathfrak{r}_c$ be an MCR of $P$.It is easy to see that $\sqrt{d}Lmax$ is the length of the longest diagonal of $\mathfrak{r}_c$. Thus the diameter of the circumscribed d-ball of $\mathfrak{r}_c$ is exactly $\sqrt{d}Lmax$. Since it is a enclosing ball and $B(\mathcal{C}_{P},\mathcal{R}_P)$ is the minimum enclosing ball, we have  $2\mathcal{R}_P\le \sqrt{d}Lmax$.

	For the other side of the inequality, suppose to the contrary that $2\mathcal{R}_P<Lmax(\mathfrak{r})$. Let $p$ and $p'$ be the pair of points that are farthest to each other in $P$, then $D(p,p')\ge Lmax(\mathfrak{r})$ since otherwise $\mathfrak{r}$ would not be the bounding rectangle of $P$. 
	Thus $2\mathcal{R}_P<Lmax(\mathfrak{r})\le D(p,p')$ and the d-ball with a diameter of $2\mathcal{R}_P$ can not enclose both $p$ and $p'$. This conflicts with the definition of MES. And thus $Lmax(\mathfrak{r})\le 2\mathcal{R}_P$ is proved.
\end{proof}

\begin{lemma}
	For any point set $P$, let $\mathfrak{r}$ be the bounding rectangle of $P$, and  $c_{\mathfrak{r}}$ and $r_{\mathfrak{r}}$ be the center and radius of the $\frac{3}{2}$-MES of $P$, then $Lmax(\mathfrak{r})\le 2\cdot r_{\mathfrak{r}}\le \frac{3}{2} \sqrt{d}\cdot Lmax(\mathfrak{r})$.	
\end{lemma}

\begin{proof}
	This lemma can be easily proved by combing Lemma \ref{lema:lmax-radius} and the definition of the $\frac{3}{2}$-MES.
	\qed
\end{proof}

Now we introduce two denotations for the following proofs.

\begin{definition}[$P(\mathfrak{r})$, the parent rectangle]
	Given a rectangle $\mathfrak{r}$ in the RST $T$, let $P(\mathfrak{r})$ be the rectangle at the parent node of $\mathfrak{r}$. If $\mathfrak{r}$ is the root then $P(\mathfrak{r})=\mathfrak{r}$.
\end{definition}

\begin{definition}[$O(\mathfrak{r})$, the outer rectangle]
	Given a rectangle $\mathfrak{r}$ in the RST $T$, define $O(\mathfrak{r})$ recursively as follows. If $\mathfrak{r}$ is the root of $T$ then $O(\mathfrak{r})$ is the minimal d-cube that contains $\mathfrak{r}$. For non root $\mathfrak{r}$, recall that when $P(\mathfrak{r})$ is split, it is split by cutting the longest side of it into two equal halves. This split in the meantime split $O(P(\mathfrak{r}))$ into two rectangles. Then let $O(\mathfrak{r})$ be the one that contains $\mathfrak{r}$.
\end{definition}

\begin{lemma}\label{lema:lmin-o-lmax-p}
$Lmin(O(\mathfrak{r}))\ge \frac{1}{2}\cdot Lmax(P(\mathfrak{r}))$.
\end{lemma}

\begin{proof}
	This lemma is an imitation of Lemma 4.1 in \cite{Callahan1995}, and thus the lemma can be proved using similar techniques. The details are omitted.
	\qed
\end{proof}

\begin{lemma}\label{lema:mappint-knbr-disjoint-cubes}
	For any rectangle $\mathfrak{r}\in H$ in the execution of Algorithm \ref{alg:all-knn}, there exists a set $S$ of d-cubes that satisfies:
	\begin{enumerate}
		\item $S$ can be mapped to $kNbr(\mathfrak{r})$ one-to-one,
		\item $Len(\mathfrak{r}')\ge \frac{1}{\sqrt{d}}r_{\mathfrak{r}_{top}}$ holds for $\forall \mathfrak{r}'\in S$, where $\mathfrak{r}_{top}$ is the top element of $H$, and
		\item the rectangles in $S$ are disjoint.
	\end{enumerate}
\end{lemma}

\begin{proof}
	For each $\mathfrak{r}'\in kNbr(\mathfrak{r})$, it is true that $Lmin(O(\mathfrak{r}'))\ge \frac{1}{2} Lmax(P(\mathfrak{r}'))$ according to Lemma \ref{lema:lmin-o-lmax-p}. Then by Lemma \ref{lema:lmax-radius} we have  $Lmax(P(\mathfrak{r}'))\ge \frac{2}{\sqrt{d}} r_{P(\mathfrak{r}')}$. On the other hand,  $r_{P(\mathfrak{r}')}\ge r_{\mathfrak{r}_{top}}$ since $\mathfrak{r}_{top}$ is the current element in $H$ and $P(\mathfrak{r})$ was popped out of $H$ before $\mathfrak{r}_{top}$. In summary, $Lmin(O(\mathfrak{r}'))\ge \frac{1}{\sqrt{d}}\cdot r_{\mathfrak{r}_{top}}$ for $\forall \mathfrak{r}'\in kNbr(\mathfrak{r})$. Thus, $O(\mathfrak{r}')$ contains a d-cube with side length $\frac{1}{\sqrt{d}}\cdot r_{\mathfrak{r}_{top}}$. Besides, it is easy to see that for arbitrary $\mathfrak{r}'$ and $\mathfrak{r}''$ in $kNbr(\mathfrak{r})$, $O(\mathfrak{r}')$ and $O(\mathfrak{r}'')$ are disjoint. 
	
	Based on the above analysis, the mapping $f$ from  $kNbr(\mathfrak{r})$ to $S$ is as follows: for each $\mathfrak{r}\in kNbr(\mathfrak{r})$, $f(\mathfrak{r}')$ is the d-cube contained in $O(\mathfrak{r}')$ and with side length $\frac{1}{\sqrt{d}}\cdot r_{\mathfrak{r}_{top}}$ described above. It can be verified that such $S$ satisfies the properties required in the lemma.
	\qed
\end{proof}

\begin{lemma}\label{lema:knbr-size}
	During the executions of Algorithm \ref{alg:compute-mes}, $|kNbr(\mathfrak{r})|= (\sqrt{d})^d$ holds for all non leaf node $\mathfrak{r}$.
\end{lemma}

\begin{proof}
	For a non leaf node $\mathfrak{r}$, $kNbr(\mathfrak{r})$ is defined to be the set of 
	$$\{\mathfrak{r}'\in H \mid D(c_{\mathfrak{r}'}+c_{\mathfrak{r}})\le r_{\mathfrak{r}'}+r_{\mathfrak{r}}+kThres(\mathfrak{r}) \},$$
	where $kThres(\mathfrak{r})=2r_{\mathfrak{r}}$. Since $\mathfrak{r}\in H$ and $kNbr(\mathfrak{r})\subseteq H$, we have $r_{\mathfrak{r}}\le r_{\mathfrak{r}_{top}}$ and $r_{\mathfrak{r}'}\le r_{\mathfrak{r}_{top}}$, where $\mathfrak{r}_{top}$ is the top element in $H$. 
	Thus it can be deduced that 
	$$kNbr(\mathfrak{r})\subseteq \{\mathfrak{r}'\in H \mid D(c_{\mathfrak{r}'},c_{\mathfrak{r}})\le 4\cdot r_{\mathfrak{r}_{top}}\}.$$ 
	Now let $S'=\{\mathfrak{r}'\in H \mid D(c_{\mathfrak{r}'},c_{\mathfrak{r}})\le 4\cdot r_{\mathfrak{r}_{top}}\} $. To establish the connections between $S'$ and the size of $kNbr(\mathfrak{r})$, the idea is to enclose $S'$ with a d-cube $\mathfrak{r}_{outer}$, and divide $\mathfrak{r}_{outer}$ into small d-cubes with side length $\frac{1}{\sqrt{d}} r_{\mathfrak{r}_{top}}$, and then Lemma \ref{lema:mappint-knbr-disjoint-cubes} can be used to derive the size of $kNbr(\mathfrak{r})$.
	
	Let $c_{\mathfrak{r}}$ to be the center of $\mathfrak{r}_{ourter}$, it can be verified that then the side length of $\mathfrak{r}_{outer}$ should be at least $10r_{\mathfrak{r}_{top}}$ to enclose $S$. 
	
	On the other hand, Lemma \ref{lema:mappint-knbr-disjoint-cubes} shows that each $\mathfrak{r}'\in kNbr(\mathfrak{r})$ can be mapped to a d-cube with side length $\frac{1}{\sqrt{d}} r_{\mathfrak{r}_{top}}$. Since the side length of $\mathfrak{r}_{outer}$ is $10r_{\mathfrak{r}_{top}}$, $\mathfrak{r}_{outer}$  contains at most $(10\sqrt{d})^d$ small d-cubes with side length of $\frac{1}{\sqrt{d}} r_{\mathfrak{r}_{top}}$. Considering the one-to-one mapping between the d-cubes with side length $\frac{1}{\sqrt{d}} r_{\mathfrak{r}_{top}}$ and the rectangles in $kNbr(\mathfrak{r})$, it can be deduced that $kNbr(\mathfrak{r})\le (10\sqrt{d})^d=O((\sqrt{d})^d)$
	\qed
\end{proof}

\begin{lemma}\label{lema:kfrd-size}
	During the executions of Algorithm \ref{alg:compute-mes}, $|kFrd(\mathfrak{r})|= O(\sqrt{d})^d$ holds for all non-leaf node $\forall \mathfrak{r}$.
\end{lemma}

\begin{proof}
	There is a similar lemma given in \cite{Vaidya1989}. The detailed proof can be found in that paper and is omitted here.
	\qed
\end{proof}

\begin{theorem}\label{thrm:time-all-knn}
	The time complexity of Algorithm \ref{alg:all-knn} is $O(k(k+(\sqrt{d})^d) \cdot n\log{n})$.
\end{theorem}

\begin{proof}
	The time to execute Algorithm \ref{alg:all-knn} consists of the following four parts: (1) the time to invoke Algorithm \ref{alg:rst}, (2) the time to invoke Algorithm \ref{alg:compute-mes}, (3) the time to manipulate the heap $H$, and (4) the time to execute Algorithm \ref{alg:new-mntn-nbrfrd}. The former two parts are already proved to be $O(dn\log{n})$ in Theorem \ref{thrm:time-alg-rst} and \ref{thrm:time-alg-mes}. According to Lemma \ref{lema:number-of-nodes-rst} there are at most $2n$ nodes in the RST $T$, so that $|H|\le 2n$. And thus the time to manipulate the heap $H$ is $O(n\log{n})$.
	
	In the rest of the proof we deal with the last part of the time complexity. First we introduce some denotations for ease of discussion. Let $Df$, $Dn$, $Ad$ and $Tr$ to be the time to execute the $DelFromFrd$, $DelFromNbr$, $AddInNbr$ and $TruncateNbr$ sub-procedures respectively. 
	Under these denotations, the time complexity of Algorithm \ref{alg:new-mntn-nbrfrd} can be expressed as follows:
	\begin{equation}\label{eqtn:mntn-time}
	|kNbr|(Df+|\mathcal{S}_{son}|\cdot Ad)+|kFrd|(Dn+|\mathcal{S}_{son}\cdot Ad)+|\mathcal{S}_{son}|^2\cdot Ad+(|\mathcal{S}_{son}|+|kFrd|)\cdot Tr
	\end{equation}
	
	According to Lemma \ref{lema:knbr-size} and \ref{lema:kfrd-size}, $|kNbr(\mathfrak{r})|$ and $|kFrd(\mathfrak{r})|$ are only bounded for non-leaf nodes. However, it can be verified that the $MntnNbrFrd$ procedure will not be invoked on leaf nodes. Thus, the $|kNbr|$ and $|kFrd|$ term can be safely up-bounded by $O((\sqrt{d})^d)$.

	It is already mentioned in the last part of Section \ref{subsec:all-knn} that $kNbr$, $kFrd$ and $Cand$ sets should be implemented by specific kind of data structures. By such implementation, $DelFromNbr$ and $TruncateNbr$ take $O(\log{|kNbr|})$ time, $AddInNbr$ take $O(\log{|Cand|}+\log{|kNbr|})$ time, and $DelFromFrd$ takes $O(\log{|kFrd|})$ time. Since these sub-procedures may be invoked on leaf nodes, and the only available bound on $|kNbr|$  and $|kFrd|$ for leaf nodes is $O(n)$, then the $Df$, $Dn$, $Tr$ can be replaced by $O(\log{n})$, and $Ad$ can be replaced by $O(\log{k}+\log{n})$.
	
	Since $\mathcal{S}_{son}\le k+1$, Equation \ref{eqtn:mntn-time} is up-bounded by the following term:
	$$(\sqrt{d})^d(\log{n}+(k+1)\log{n})\times 2+(k+1)^2\log{n}+(k+1+(\sqrt{d})^d)\log{n},$$
	which is $O(((k(\sqrt{d})^d)+(k+1)^2)   \log{n})$.

	On the other hand, the $while$ loop in Algorithm \ref{alg:all-knn} can be executed for at most $n$ times, since there are at most $2n$ nodes in $T$ and each execution of the $while$ loop visits $2$ nodes.
	Then the time to execute Algorithm \ref{alg:mntn-nbr-frd} is $O(k(k+(\sqrt{d})^d) \cdot n\log{n})$.
	
	Finally, adding the four parts of time and the time complexity of Algorithm \ref{alg:all-knn} is proved to be $O(k(k+(\sqrt{d})^d) \cdot n\log{n})$.
	\qed
\end{proof}

\begin{theorem}\label{thrm:space-all-knn}
	The space complexity of Algorithm \ref{alg:all-knn} is $O(n^2)$.
\end{theorem}

\begin{proof}
	Only the rectangles in $H$ will store their $kNbr$ and $kFrd$ sets. Thus the space complexity of Algorithm \ref{alg:all-knn} can be computed by multiplying the size of $kNbr$ and $kFrd$ sets and the size of $H$, which is $O(n^2)$.
\end{proof}

In the proof of Theorem \ref{thrm:time-all-knn} and \ref{thrm:space-all-knn} we use $O(n)$ as the bound on $|kNbr(\mathfrak{r})|$ for leaf node $\mathfrak{r}$, which is sufficient to prove the desired $O(n\log{n})$ bound on the running time. Though we can not give an tighter upper bound of the size of $kNbr(\mathfrak{r})$ for leaf node $\mathfrak{r}$, we have the following lemma depicting the structure of it.

\begin{lemma}
	Let $\mathfrak{r}$ to be a leaf node in $T$, and let $\mathfrak{r}'$ be an arbitrary rectangle in $kNbr(\mathfrak{r})\setminus Cand(\mathfrak{r})$. Then any $\mathfrak{r}''\in kNbr(\mathfrak{r}\setminus Cand(\mathfrak{r}))$ must satisfy $d(c_{\mathfrak{r}''},c_{\mathfrak{r}'})\le d(c_{\mathfrak{r}}, c_{\mathfrak{r}'} )+\frac{5}{2}r_{\mathfrak{r}_{top}}$, where $\mathfrak{r}_{top}$ is the top element in the heap $H$.
\end{lemma}

\begin{proof}
	Since $\mathfrak{r}'\notin Cand(\mathfrak{r})$, $$Qlty(\mathfrak{r}',\mathfrak{r})=\sqrt{D(c_{\mathfrak{r}'},c_{\mathfrak{r}}+\frac{1}{2}r_{\mathfrak{r}'})^2+r_{\mathfrak{r}'^2}}>kThres(\mathfrak{r}).$$ 
	Next, since $\mathfrak{r}_{top}$ is the top element in the heap $H$ and $\mathfrak{r}'$ is an element in $H$, we have
	$$kThres(\mathfrak{r})< \sqrt{D(c_{\mathfrak{r}'},c_{\mathfrak{r}}+\frac{1}{2}r_{\mathfrak{r}'})^2+r_{\mathfrak{r}'^2}}< D(c_{\mathfrak{r}'},c_{\mathfrak{r}})+\frac{1}{2}r_{\mathfrak{r}'}+r_{\mathfrak{r}'}< D(c_{\mathfrak{r}'},c_{\mathfrak{r}})+\frac{3}{2}r_{\mathfrak{r}_{top}}.$$ 
	On the other hand, 
	$D(c_{\mathfrak{r}''},c_{\mathfrak{r}})\le r_{\mathfrak{r}''}+r_{\mathfrak{r}}+kThres(\mathfrak{r})=r_{\mathfrak{r}''}+kThres(\mathfrak{r})$ since $\mathfrak{r}''\in kNbr$ and $r_{\mathfrak{r}}=0$. Finally, we have 
	$$D(c_{\mathfrak{r}''},c_{\mathfrak{r}})\le r_{\mathfrak{r}''}+kThres(\mathfrak{r})\le r_{\mathfrak{r}_{top}}+D(c_{\mathfrak{r}'},c_{\mathfrak{r}})+\frac{3}{2}r_{\mathfrak{r}_{top}} =D(c_{\mathfrak{r}'},c_{\mathfrak{r}})+\frac{5}{2}r_{\mathfrak{r}_{top}}.$$
	which proves the desired result.
	\qed
	
\end{proof}

\section{Conclusion}\label{sec:conc}
In this paper the All-k-Nearest-Neighbors problem is considered. An algorithm proposed in \cite{Vaidya1989} is the inspiring work of this paper. In \cite{Vaidya1989} the author claimed  that the algorithm has an upper bound of $O(d^dn\log{n})$ on the time complexity. However, we find that this bound is unachievable according to the descriptions in \cite{Vaidya1989}. We give formal analysis that the algorithm needs at least $\Omega(n^2)$ time. On the other hand, we propose another algorithm for the All-k-Nearest-Neighbor problem whose time complexity is truly up-bounded by $O(k(k+d^d)n\log{n})$. After all, we have renewed an result that has an history of over 30 years and has been cited mor than 300 times. Considering the importance of the All-k-Nearest-Neighbor problem, this work should be considered valuable.

\bibliographystyle{splncs04}
\bibliography{library}
\end{document}